\documentclass[aps,pre,reprint]{revtex4-1}

\usepackage{amsmath}
\usepackage{amsfonts}
\usepackage{amsthm}
\newtheorem{lemma}{Lemma}
\usepackage{bm}
\usepackage{graphicx}

\newcommand{\Ref}[1]{Ref.~\cite{#1}}
\newcommand{\Refs}[1]{Refs.~\cite{#1}}
\newcommand{\SEC}[2]{\section{\label{sec:#1}#2}}
\newcommand{\Sec}[1]{Sec.~\ref{sec:#1}}
\newcommand{\FIG}[2]{\caption{\label{fig:#1}#2}}
\newcommand{\Fig}[1]{Fig.~\ref{fig:#1}}
\newcommand{\EQ}[1]{\label{eq:#1}}
\newcommand{\Eq}[1]{Eq.~(\ref{eq:#1})}

\newcommand{\prm}[1]{#1^\prime}

\newcommand{\avg}[1]{\langle#1\rangle}

\newcommand{\var}[1]{\mathrm{var}\{#1\}}

\newcommand{\skw}[1]{\mathrm{skw}\{#1\}}

\newcommand{\Li}{\mathrm{Li}}

\begin{document}

\title{The Langevin equation for systems with a preferred spatial direction}
\date{\today}
\author{Roman Belousov}\email{belousov.roman@gmail.com}
\affiliation{The Rockefeller University, New York 10065, USA}
\author{E.G.D. Cohen}\email{egdc@mail.rockefeller.edu}
\affiliation{The Rockefeller University, New York 10065, USA}
\affiliation{Department of Physics and Astronomy, The University of Iowa, Iowa
	City, Iowa 52242, USA}
\author{Lamberto Rondoni}\email{lamberto.rondoni@polito.it}
\affiliation{Dipartimento di Scienze Matematiche and Graphene@Polito Lab
Politecnico di Torino - Corso Duca degli Abruzzi 24, 10125, Torino, Italy}
\affiliation{INFN, Sezione di Torino - Via P. Giuria 1, 10125, Torino, Italy}
\begin{abstract}
  In this paper, we generalize the theory of Brownian motion and the Onsager-Machlup
  theory of fluctuations for spatially symmetric systems to equilibrium and nonequilibrium
	steady-state systems with a preferred spatial direction, due to an external force.
	To do this, we extend the Langevin equation to include a bias, which is introduced
	by the external force and alters the Gaussian structure of the system's fluctuations.
	By solving this extended equation, we demonstrate that the statistical properties
	of the fluctuations in these systems can be predicted from physical observables,
	such as the temperature and the hydrodynamic gradients.
\end{abstract}
\maketitle

\SEC{intro}{Introduction}
The dynamical theory of fluctuations in physical systems began to assume its modern
form with the seminal papers of Onsager and Machlup \cite{OM1953,MO1953}. They
proposed to describe the time evolution of the thermodynamic fluctuating quantities,
as well as of the hydrodynamic and electrodynamic variables, by a stochastic Langevin
equation \cite[Chapters 1-2]{Coffey2012Langevin}. Originally Onsager and Machlup
considered fluctuations only in equilibrium systems. Generalizations of the Langevin
equation for fluctuations to nonequilibrium steady states followed, {\it e.g.}
\Refs{KSSH2015I,KSSH2015II,MorgadoQ2016}, as discussed later in this paper.

The formalism of the Langevin equation was first developed in the theory of Brownian
motion \cite[Chapters 1-2]{Coffey2012Langevin}. Later Onsager and Machlup proposed
\cite{OM1953,MO1953} that the fluctuations of the thermodynamic quantities can be
described by the {\it same} stochastic equation, as used for the velocity fluctuations
of a Brownian particle in an equilibrium system. That is, the time evolution of
a fluctuating quantity $\alpha(t)$ obeys the following Langevin dynamics \footnote{
	The differential equation~(\ref{eq:dff}) is of the first order with respect to
	time. A second order version of the Langevin equation was considered in \Ref{MO1953}
	for systems, in which the fluctuations of the currents should be taken into
	account. This modifies merely the deterministic character of the resulting dynamics,
	while the steady-state probability of the fluctuations, studied in this paper,
	remain unchanged, {\it cf} \cite[Section II.3]{Chandrasekhar1943}.
}:
\begin{eqnarray}\EQ{dff}
  d\alpha(t) = -A \alpha(t) dt + B dW(t)\text{.}
\end{eqnarray}
Here $A$ and $B$ are positive constants, whose values and physical interpretation
depends on the system under consideration, while $t$ is the time and $dW(t)$ is
a white noise, defined as a differential of a Wiener process $W(t)$
\cite[Chapter 1]{Coffey2012Langevin}:
\begin{equation}\EQ{dif}
  W(t) = \int_0^t dW(\prm{t})\text{.}
\end{equation}

The first term on the right hand side of \Eq{dff} is a damping force with a
friction constant $A$, which ensures that the fluctuations of the quantity
$\alpha(t)$ decay to a macroscopically observable average value $\avg{\alpha(t)}$.
The second term, $B dW(t)$, represents physically a microscopic noise of constant
intensity $B$. It has a Gaussian nature, since $W(t)$ in \Eq{dif} is by definition
a normally distributed random variable of zero mean and variance $t$.

By solving \Eq{dff}, Onsager and Machlup predicted a Gaussian structure of the
fluctuations in equilibrium systems. However, they explicitly omitted in their
treatment \cite{OM1953} rotating systems and systems subject to an external
field, because these do not possess the property of microscopic reversibility.

In this paper, we will treat the dynamical theory of fluctuations for a class of
systems subject to an external field. This includes not only equilibrium systems
in an external potential, such as a gravitational potential, but in addition systems
maintained in a nonequilibrium steady state by an external thermodynamic, hydrodynamic
or electrodynamic gradient.

Indeed, recent studies confirm a non-Gaussian structure of fluctuations in this
class of systems \cite{KSSH2015I,Gustavsson_2006,UtsumiSaito2009,PRE2016}. In
particular, \Refs{Gustavsson_2006,UtsumiSaito2009,PRE2016} report that the
probability distribution of their fluctuations becomes asymmetric and, thus, acquires
{\it a skewness} \footnote{
  Skewness is related to the third moment of a probability distribution, so that
  symmetric distributions, like the Gaussian, have zero skewness.
}.

The above mentioned theoretical and experimental studies indicate that the probability
distribution of fluctuations is biased in the presence of a {\it preferred spatial
direction}, which is induced by an externally applied force. In contrast to such
systems, the original Langevin equation \Eq{dff} has a peculiar symmetry, since
it has no preferred spatial direction. For, it assigns equal probabilities to both
positive and negative fluctuations of $\alpha(t)$, {\it i.e.} neither positive,
nor negative fluctuations are favored. However, this symmetry is broken, when an
external field introduces a special direction and, as conjectured in \cite{UtsumiSaito2009,PRE2016}, alters
the microscopic noise in this class of systems. A consequence of this is a non-Gaussian
structure of their fluctuations.

This symmetry argument can be introduced formally using the principle of dissymmetry
due to Curie \cite{Shubnikov1988,Curie1894}. In the treatment of Onsager and Machlup
\cite{OM1953} it was implicit, that the systems they considered belong to Curie's
limiting point group of the highest symmetry $\infty / \infty \cdot m$ \cite{Shubnikov1988}.
A skewness of the fluctuations was observed in the systems, which lack some symmetry
operations with respect to this point group. In all these cases the bias of the
fluctuations is evidently due to a reduction of symmetry or, as introduced by Curie,
due to a dissymmetry \footnote{
  ``C'est la dissym\'{e}trie qui cr\'{e}e le ph\'{e}nom\`{e}ne'' (It is the dissymetry,
	which creates the phenomenon) \cite{Curie1894}.
} with respect to the systems regarded by Onsager and Machlup. In this paper, we
will develop a Langevin equation for the class of physical systems with a polar
direction \cite{Shubnikov1988}, which is due to the external force.

We emphasize the role of the spatial asymmetry, in contrast to the temporal asymmetry
of microscopically irreversible systems, which are dealt with by the Microscopic
Fluctuation Theory \footnote{
  To reverse the evolution of such systems, the sign of the external force should
  be changed together with that of the velocities and of the time \cite{MFT2015}.
} \cite{MFT2015}. It was long thought, that macroscopic irreversibility would be
an immanent property of all systems in an external potential. However, it was shown,
that the presence of a magnetic field does not change the time-reversal symmetry of
an equilibrium system \cite{Bonella2014}, while altering the probability structure
of its fluctuations \cite{UtsumiSaito2009}, as discussed earlier. In \Sec{end},
we will remark, though, that a magnetic field may not belong to the class of
systems, which are liable to the theoretical arguments of this paper.

Recently it was shown in \Ref{KSSH2015I}, that the original Langevin equation can
be extended by adding a third stochastic term, which acts as an external force and
causes, together with the white noise, a non-Gaussian behavior of the fluctuations
in a nonequilibrium system. To make further progress, the authors of \Refs{KSSH2015II,MorgadoQ2016}
{\it assumed}, that this term is a Poisson process, and added it to \Eq{dff},
which then reads
\begin{eqnarray}\EQ{pois}
  d\alpha(t) = -A \alpha(t) dt + B dW(t) + P(t)\text{.}
\end{eqnarray}
Here $P(t)$ is the Poisson process, also referred to as a shot noise \cite[Chapter 6]{Stratonovich},
which has a constant rate parameter and, in general, a variable intensity parameter
\cite{MorgadoQ2016}.

The Poisson process assigns a non-zero probability only to non-negative numbers,
so that the role it plays in \Eq{pois} is two-fold. First, it acts as an external
force and, second, it introduces a bias, which makes \Eq{pois} consistent with the
symmetry of the class of systems considered here. Also, the microscopic noise is
not represented solely by the white noise, but has an additional contribution due
to the stochastic nature of the third term, $P(t)$.

Apparently the shot noise in \Eq{pois} was motivated by its applications in the
theory of electric conductance \cite{KSSH2015II,BlanterB2000}. The Poisson process
is discrete and makes \Eq{pois} singular, {\it cf.} \cite{MorgadoQ2016}. Although
in the theory of electric conductance this singularity is explained by the discrete
nature of the electric charge \cite{BlanterB2000}, it is a rather curious aspect
of \Eq{pois} for a Langevin dynamics in the context of classical statistical mechanics.

In this paper we propose to replace the shot noise in the Ansatz of
\Refs{KSSH2015II,MorgadoQ2016} by a different non-Gaussian stochastic term, so that
the extended Langevin equation for the fluctuations in systems with a preferred
spatial direction would read:
\begin{eqnarray}\EQ{ours}
  d\alpha(t) = -A \alpha(t) dt + B dW(t) + C dE_\tau(t)\text{.}
\end{eqnarray}
Here $C$ is a positive or negative constant, while $dE_\tau(t)$ is a time
differential of a Gamma process $E_\tau(t)$ with a time scale parameter $\tau$ \footnote{
  The Gamma process is characterized by statistically independent increments, each
  having a Gamma probability distribution \cite[Chapter I]{FrenkNicolai2007,Dufresne1991,SteutelvHarn}.
}, {\it cf.} \cite[Chapter I]{FrenkNicolai2007,Dufresne1991,SteutelvHarn}.
We will call $dE_\tau(t)$ an {\it exponential noise} for a reason, clarified in
\Sec{gamma}.

The first improvement achieved by \Eq{ours}, with respect to \Eq{pois}, is its
statistical foundation, which is comparable to that of the original Langevin
equation. For, unlike the shot noise assumed in \Eq{pois}, both the white noise
and the exponential noise in \Eq{ours} can be deduced from simplified models of
the physical systems studied in this paper. In fact, Chandrasekhar \cite[Chapter I]{Chandrasekhar1943}
considered a discrete physical model of microscopic noise and obtained the Wiener
process as a continuous limit of a simple symmetric random walk. We adapt the same
approach here, by modeling the effect of an external force with an asymmetric random
walk, which in a similar continuous limit leads to the concept of exponential noise.
To complete the analogy with Chandrasekhar's method, we will verify in \Sec{gamma}
that, like the Wiener process, the Gamma process also arises in a more elaborate
model of a random flight.

The second advantage of \Eq{ours} is that, since the exponential noise is not singular,
in contrast to the shot noise, it fits more naturally into a stochastic differential
equation. While the Poisson process is discrete, it has a highly non-trivial continuous
counterpart \cite{Ilienko2013}, which is, nonetheless, not considered by the proponents
of \Eq{pois}. As mentioned earlier, the discrete nature of the third term in \Eq{pois}
introduces a singularity. In contrast to this, the Gamma process, like the Wiener
process, is non-singular and assumes a simple mathematical expression in both
continuous and discrete stochastic dynamics. Therefore the theory and the treatment
of the Langevin equation \Eq{ours} is in principle simpler than that of \Eq{pois}.

In \Sec{new} we will show that the statistical properties of the fluctuating
quantity $\alpha(t)$, which evolves according to the extended Langevin equation~(\ref{eq:ours}),
can be computed in terms of the same physical parameters, which characterize the
macroscopic state of the systems, studied in this paper. In particular, we will
confirm the non-Gaussian structure of the fluctuations by calculating their skewness.
Moreover, the {\it sign} of the skewness will depend on the external force in a
manner, which was already observed by an earlier experiment \cite{PRE2016}.

Finally, we note that, while the behavior exhibited by \Eq{pois} is qualitatively
very similar to that of \Eq{ours}, they differ in principle. Equation~(\ref{eq:pois})
may be applicable to some systems, which are listed in \Ref{MorgadoQ2016} and which
need a noise term of a discrete nature, {\it e.g.} systems of a small size. Nonetheless,
in this paper we argue that \Eq{ours} will find a broader range of applications
for a variety of physical systems considered by classical statistical mechanics.

\SEC{xmpl}{A simplified physical example}
To provide a physical insight into the dynamics, described by a Langevin
equation of the form \Eq{pois} or \Eq{ours}, we consider in this section a
macroscopic system as an idealization of the systems studied by classical statistical
mechanics, which are of interest in this paper. This will allow us to develop a
decomposition of the random noise into two parts: a symmetric and asymmetric random
processes, respectively. The latter will also incorporate the action of an external
field. As discussed afterwards, such a decomposition is not obvious at the level
of classical statistical mechanics, but it is much clearer in the example considered
below or some biological systems.

First, consider a man in a boat on a lake. When the man just sits in the boat, the
motion of the boat can be described by the Langevin equation (\ref{eq:dff}), where
the damping force would be due to the friction of the boat in the water and the
white noise would be caused by spontaneous fluctuations due to the waves on the
water surface and the wind blows. The stochastic term is motivated by the symmetry
of this physical system, which {\it a priori} does not favor any direction of motion,
so that the excitations pushing the boat forward or backward are equally probable.
As a result, the boat's velocity is distributed symmetrically around zero.


Now imagine, that the man begins to paddle, so that the boat is propelled forward,
by impulses, which are imparted by the oar at a certain rate. This rate will depend
on the rowing rhythm, which is, in general, irregular. For instance, the man
sometimes may row slower and other times faster. This irregularity of the rowing
rhythm can be accounted for statistically, if we regard the total force imparted
by the rower to the boat as a {\it random variable}, which has some definite average
value over a sufficiently long time interval and assumes {\it only non-negative
values}. This random variable, when added to the original Langevin equation \Eq{dff}
as a third term, yields a stochastic dynamics of the form \Eq{pois} or \Eq{ours}.

This new stochastic term, which represents an external force acting on the boat,
has one important attribute, which distinguishes it from the white noise term
discussed earlier. Namely, the external force assumes {\it only non-negative}
values, since the rower propels the boat always forward. Clearly, the average
velocity of the boat will then be positive. However, due to the external force,
the fluctuations of the velocity are amplified in the forward direction and suppressed in the backward
direction. This introduces a {\it bias} for the forward fluctuations of the boat's
velocity and, thus, reduces the symmetry of the system.

Here it is relevant to remark, that if the third term in \Eq{pois} or in \Eq{ours}
were either a constant or another Wiener process, the resulting fluctuations would
have a Gaussian structure. In fact, a certain change of variables would then transform
these equations into the form of \Eq{dff}. Therefore, both the stochastic nature
and the absence of negative values of the external force turn out to be crucial
to reproduce the non-Gaussian nature of the fluctuations in the class of systems
considered in this paper.

Generalizing the above argument, we will assume that in {\it all} physical systems
of interest for this paper, the random noise can be represented as a linear superposition
of a symmetric term, being the white noise, and {\it some} asymmetric term, which
corresponds to the external force. The latter is asymmetric, because it never takes
on negative values. Both models discussed in the Introduction, \Eq{pois} and \Eq{ours},
are constructed in this way.

The described decomposition of the random noise will be assumed, inspite of the
fact, that in a real thermodynamic, hydrodynamic or electrodynamic system the microscopic
noise and the external force can not be easily separated. For example, a Brownian
particle, which collides with the molecules of a fluid subject to a density gradient,
will drift, on average, in a certain direction. Then, since both the microscopic
noise and the external force acting on the Brownian particle are both caused by
the collisions with the fluid molecules, it is not obvious that each of the two
can be represented in the Langevin equation by a distinct separate term of stochastic
nature.

Nonetheless, some biological systems \cite{biology}, which are traditionally modeled
by stochastic dynamics, bear some similarity to the rower example. For instance,
a bacterium, swimming in a liquid by moving its flagellum, is an obvious parallel
with the man paddling a boat.

\SEC{gauss}{White noise}
This section reviews a simple 1-Dimensional (1D) random walk, as it was used by
Chandrasekhar \cite[Chapter I]{Chandrasekhar1943} to motivate the white noise term
for the Langevin dynamics described by Eqs.~(\ref{eq:dff}), (\ref{eq:pois}) and
(\ref{eq:ours}). In a slightly modified form, the same approach will be adopted
in the next section to deduce the form of the third term in \Eq{ours}.

Consider a particle, which suffers displacements along a line in the form of
discrete steps of equal length. The particle moves one step forward with
probability $p(1) = 1/2$, while the probability of a backward step is
$p(-1) = 1-p(1) = 1/2$. Equal probabilities of backward and forward displacements
do not favor any direction of the motion. This is consistent with the symmetry
of the system, described in \Sec{xmpl}, where a man sits in a boat without doing
anything.

The problem is to find the probability $W_N(m)$, that the particle has moved to
a point $m$ after a series of $N$ steps, $-N \le m \le N$. Without loss of generality,
we assume that the initial position of the particle is at zero $m_0=0$, so that
the total displacement $\Delta{m} = m - m_0 = m$ equals the final position of the
particle. The exact solution is given by the binomial distribution
\cite[Chapter I]{Chandrasekhar1943}:
\begin{equation}\EQ{wnm}
  W_N(m) = \frac{N! [p(1)]^{(N+m)/2} [p(-1)]^{(N-m)/2}}{[(N+m)/2]! [(N-m)/2]!}
  \text{.}
\end{equation}

As can be shown \cite[Chapter I]{Chandrasekhar1943}, the binomial distribution
\Eq{wnm} with $p(1) = 1/2$ approaches asymptotically a Gaussian \footnote{
	The Gaussian approximation $p_G(m)$ is accurate only around the mean value of
	$m$, {\it cf.} \cite{Keller2004}. Nonetheless, the original theory of Langevin
	equation is not concerned with corrections for the large deviations from the mean,
	which have vanishingly small probabilities.
} $p_G(m)$:
\begin{eqnarray}\EQ{nrm}
  W_N(m) \underset{N\to\infty}{\to} p_G(m) = (2 \pi N)^{-1/2}
    \exp\left( -\frac{m^2}{2 N} \right)\text{.}
\end{eqnarray}

To obtain the continuous limit of \Eq{nrm}, one introduces a density of sites
accessible to the particle per unit length $\rho = \Delta{m}/\Delta{x}$ and the
rate of displacements suffered per unit time $\nu=\Delta{N}/\Delta{t}$, where $\Delta{x}$
and $\Delta{t}$ are now, respectively, the continuous increments of coordinate and
time. Then, using \Eq{nrm} for $\rho$ and $\nu$ fixed in the limit $\Delta{x}\to0$
and $\Delta{t}\to0$, one finds from \Eq{nrm} the probability density of particle's
displacement $\Delta{x}$ within a time interval $\Delta{t}$ \cite[Chapter I]{Chandrasekhar1943}:
\begin{eqnarray}\EQ{one}
  p_G(\Delta{x}, \Delta{t}) = \frac{1}{\sqrt{4 B^2 \Delta{t}}}
      \exp\left(-\frac{\Delta{x}^2}{4 B^2 \Delta{t}} \right)\text{,}
\end{eqnarray}
where $B^2 = \nu/(2 \rho^2)$.

The coordinate $\Delta{x}$ is thus a Gaussian random variable. Therefore the continuous
limit, used to obtain \Eq{one}, can be interpreted in terms of the Wiener process,
{\it cf.} \cite[Chapter II Lemma I]{Chandrasekhar1943}, which allows us then to
pose that:
\begin{eqnarray}\EQ{two}
  \Delta{x} = B W(\Delta{t}) = B \int_0^{\Delta{t}} dW(t)
\end{eqnarray}

Instead of random displacements in the coordinate space, one can consider ``displacements''
in a velocity space, as in the problem of Brownian motion. This way one obtains
the white noise in the Langevin equation \Eq{dff}.

The equal length of each step in this simple random walk problem turns out to be
insignificant, as shown in Ref.~\cite[Chapter I]{Chandrasekhar1943}. In particular,
random flight models, where the size of each step is sampled from a variety of
probability distributions, lead again to a Gaussian distribution of the particle's
total displacement. The key aspect, therefore, is that the considered dynamics favors
no particular direction of motion, since it assigns equal probabilities to the
forward and backward displacements at each step.

\SEC{gamma}{Exponential noise}
To motivate the third term of the Langevin dynamics \Eq{ours}, we need to exclude
negative values of the external force it represents, as was suggested in \Sec{xmpl}.
This constraint can be implemented in the model of a simple random walk, reviewed
in the preceding section, by a minor modification. Namely, the particle now will
make {\it only forward} steps with the same probability $p(1)=1/2$ or it will
{\it remain at rest} with the probability $p(0)= 1 - p(1) =1/2$. Then the particle's
position can take on values in the integer range $0 \le m \le N$.

As in \Sec{gauss}, the problem is again to find the probability $W_N(m)$ that
the particle moved from its initial position at zero to the position $m = \Delta{m}$
after $N$ steps. The exact solution is again a binomial distribution, which can
be obtained by the same argument as \Eq{wnm} \cite[Chapter I]{Chandrasekhar1943}:
\begin{equation}\EQ{binomial}
  W_N(m) = \frac{N! [p(1)]^m [p(0)]^{N-m} }{m! (N-m)!}\text{.}
\end{equation}

As was done in \Sec{gauss}, the binomial probability mass \footnote{
  The probability mass function is the discrete analogue of the probability density
  function for continuous random variables.
} function \Eq{binomial} can be approximated by a Gaussian, which has a support
$(-\infty,\infty)$. However, we emphasized earlier, that the zero probability of
all negative values has a {\it physical} significance, because it acts as an
external force always acting in the forward direction. For that reason we have to
abandon the Gaussian approximation, which holds only for small deviations from the
mean, {\it cf.} \cite{Keller2004}.

Instead of the Gaussian approximation, it would be tempting to resort to the Poisson
distribution, which is traditionally used as a limiting case of the binomial
distribution \Eq{binomial} for $p(1)\to0$ \cite[Section 3.6]{Good1986,Durrett},
when a random variable of interest has its support on the half real line $[0, \infty)$.
In fact, this approximation leads to the shot noise term in \Eq{pois}, proposed
by \Refs{KSSH2015II,MorgadoQ2016}. However, in accordance with the ideas developed
in \Sec{xmpl}, the Poisson model is restricted to weak external forces, because
it requires a vanishing probability of forward displacements, so that the particle
mostly stays where it is. Fortunately, this rather restrictive assumption is irrelevant
for the case $p(1)=1/2$ of interest here, which actually is much better approximated by a different expression, as
follows.

For the special case $p(1)=1/2$, which is of interest here, we will demonstrate
that the binomial distribution \Eq{binomial} can be approximated by a Gamma
distribution $p_\Gamma$ \cite[Chapter 15]{Krishnamoorthy} in the limit $N\to\infty$:
\begin{equation}\EQ{gamma}
  W_N(m) \underset{N\to\infty}{\to} p_\Gamma = \frac{m^{N-1}}{\theta^N\Gamma(N)}
    \exp(-m/\theta)\text{,}
\end{equation}
with an average value $\avg{m} = N \theta$, a variance $\var{m} = N \theta^2$
and the parameter $\theta = p(1) = p(0) = 1/2$, which is the mean rate of forward
moves per step.

To the best of our knowledge, this work is the first to propose the Gamma
approximation of the binomial distribution, which is motivated by the fact that
the mean and the variance of $W_N(m)$ in \Eq{binomial} \cite[Chapter 3]{Krishnamoorthy}
coincide with those of $p_\Gamma$ \cite[Chapter 15]{Krishnamoorthy}. While a formal
mathematical argument is given in Appendix~\ref{sec:approx}, below we illustrate
the efficiency of \Eq{gamma} by the numerical simulations in \Fig{sim}. The Gamma
approximation becomes indistinguishable from a Gaussian for a sufficiently large $N$,
like in \Fig{sim}(b). An excellent agreement between the histograms and the Gamma
probability distribution is evident for increasing $N$ in \Fig{sim}, while the Poisson
distribution gives a poor representation of the simulation data, as expected for
a non-vanishing probability of the forward step $p(1)=1/2$.

\begin{figure*}[!ht]
\includegraphics[width=2\columnwidth]{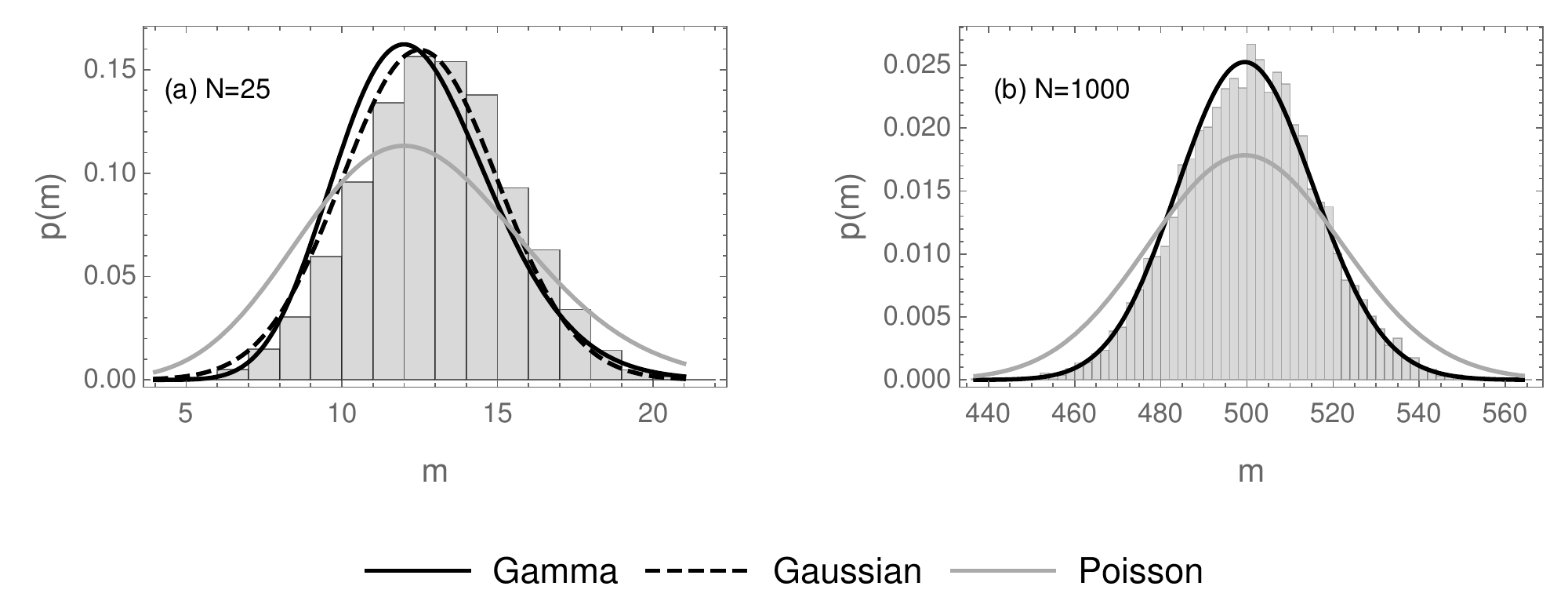}
\FIG{sim}{
  Histograms of two random walk simulations: (a) $N=25$, (b) $N=1000$. In (b)
  the Gaussian approximation is not plotted, because it is indistinguishable
  from the graph of the Gamma distribution when $N$ is so large. The Poisson
  distribution provides a poor approximation of the histogram data, since the
  probability of the forward step is not sufficiently small.
}
\end{figure*}

To specify the continuous counterparts of the discrete variables $m$ and $N$ in
\Eq{gamma}, we express the average displacement $\avg{\Delta{m}}$ in terms of the
displacement rate $\nu = \Delta{N}/\Delta{t}$ per unit time and the density of
positions $\rho = \Delta{m} / \Delta{x}$ per unit length, introduced in \Sec{gauss},
so that
\begin{equation}\EQ{spec}
  \avg{\Delta{x}} = \avg{\Delta{m}/\rho} = \Delta{N} \theta / \rho = C \Delta{t}/\tau \text{,}
\end{equation}
where $C = \theta / \rho$ and $\tau = 1/\nu$.

If we fix $\rho$ and $\nu$ for $\Delta{x}\to0$ and $\Delta{t}\to0$ as in \Sec{gauss},
the continuous limit of \Eq{gamma} then follows from the property of infinite
divisibility of the Gamma distribution \cite[Chapter I]{SteutelvHarn}. This means,
in particular, that \Eq{gamma} can be represented as a sum of $N$ independent random
variables distributed according to an {\it exponential law} of intensity $\theta$
\cite[Chapter 15]{Krishnamoorthy}. This property and \Eq{spec} both motivate the
replacing the sum over $N$ in the continuous limit by a time integral of the
{\it exponential noise} $dE_\tau(t)$, which is then defined as the differential of the
Gamma process $E_\tau(t)$, {\it cf.} \cite{FrenkNicolai2007,SteutelvHarn}, so that
$\Delta{x}$ obeys the probability law of a Gamma distribution:
\begin{eqnarray}
  \EQ{a1}
  p(\Delta{x},\Delta{t}) &=& \frac{\Delta{x}^{\Delta{t}/\tau-1}}{C^{\Delta{t}/\tau}
    \Gamma(\Delta{t}/\tau)} \exp\left(
      -\frac{\Delta{x}}{C}\right
      )
  \\\EQ{a2}
  \Delta{x} &=& C E(\Delta{t}) = C \int_0^{\Delta{t}} dE_\tau(t)
\text{.}
\end{eqnarray}
Here \Eq{a1} and \Eq{a2} define the properties of the exponential noise in the
same manner as \Eq{one} and \Eq{two}, respectively, determine the properties of
the white noise.

We concluded \Sec{gauss} by pointing out, that the white noise also emerges in models
of a random flight, where the length of the particle's displacements is sampled
at each step from a symmetric probability distribution. Similarly, the random walk
considered in this section can be generalized to a random flight with steps of a
variable length. If the length of each displacement is sampled from an exponential
probability distribution, the Gamma distribution of the particle's total displacement
arises again as the sum of independent exponentially distributed random variables,
a property already used above to deduce \Eq{a1}.

The analogy between the random walk problems of this and of the previous section
is now complete. In \Sec{gauss} the Wiener process was obtained as the continuous
limit of the symmetric random walk problem. By a similar argument, above we deduced
the Gamma process from the continuous limit of an asymmetric random walk.

Finally, we conjecture that the generalization of the Langevin equation \Eq{dff}
to systems with a preferred spatial direction, induced by an external force, is
given by \Eq{ours}. The new third term of that equation, {\it i.e.} $C dE_\tau(t)$, is
proportional to the exponential noise, defined by Eqs.~(\ref{eq:a1})~and~(\ref{eq:a2}).
If instead of the coordinate space we considered the velocity space of a Brownian
particle, $C/\tau$ would have a physical meaning of the mean external force, as
discussed in the next section.

\SEC{new}{Solution of the extended Langevin equation}
The extended Langevin equation (\ref{eq:ours}), can be solved by a straightforward
generalization of the method used in Ref.~\cite[Chapter II]{Chandrasekhar1943}
for \Eq{dff}. To do this, we first denote by $\epsilon(t)$ the sum of the folowing
two stochastic terms
\begin{equation}\EQ{super}
  \epsilon(t) = B dW(t)/dt + C dE_\tau(t)/dt\text{,}
\end{equation}
so that \Eq{ours} can be rewritten as
\begin{equation}\EQ{langevin}
  d\alpha(t)/dt = -A \alpha(t) + \epsilon(t)\text{.}
\end{equation}

A formal solution of \Eq{langevin} was already given in Ref.~\cite[Chapter II]{Chandrasekhar1943},
which we repeat here in our notation:
\begin{equation}\EQ{sln}
  \alpha(t) = \alpha_0 \exp(- A t) + \exp(-A t) \int_0^t ds \exp(A s) \epsilon(s)\text{,}
\end{equation}
where $\alpha_0 = \alpha(0)$ is an initial value condition.

Since in this paper we do not need a solution of \Eq{langevin} for a particular
physical system, we will focus our attention on the Steady-State (SS) solution
$\alpha_\mathrm{SS}$. This will suffice for our interest in the statistical nature
of the fluctuations described by \Eq{ours}, as wa anticipated in Introduction.
Taking the steady-state limit of \Eq{sln} we have
\begin{eqnarray}\EQ{ss}
  \alpha_\mathrm{SS} = \lim_{t\to\infty} \alpha(t) = \lim_{t\to\infty} \left[
    \int_0^t ds \exp[A (s-t)] \epsilon(s)\right
  ]\text{.}\nonumber\\
\end{eqnarray}

The decomposition of $\epsilon(t)$ in \Eq{super} splits the integral on the right
hand side of \Eq{ss} into a sum of two terms:
\begin{equation}\EQ{split}
  B \int_0^t \exp[A (s-t)] dW(s) + C \int_0^t \exp[A (s-t)] dE_\tau(s) \text{.}
\end{equation}

The first integral in \Eq{split} is given by Lemma I of Ref.~\cite[Chapter II]{Chandrasekhar1943}.
The result of integration is a normally distributed random variable $\beta(t)$,
with a zero mean and a variance
$$\var{\beta(t)} = \frac{B^2}{2 A} [1-\exp(-2 A t)]\text{,}$$
which in the steady-state limit becomes
$$\lim_{t\to\infty}\var{\beta(t)}= \frac{B^2}{2 A}\text{.}$$

For the second integral in \Eq{split} we need another result, which is analogous
to the above cited Lemma I of Ref.~\cite[Chapter II]{Chandrasekhar1943} but for
the exponential noise, is established by Lemma~\ref{lemma} in Appendix~\ref{sec:lemma}.
There we show, that the second integral in \Eq{split} is a random variable, given
by a Gamma-mixture distribution, and compute its mean, variance and skewness.
From now on we will denote this random variable by $\gamma(t)$.

In summary, we found that $\alpha(t)$ is a sum of two independent random variables,
a Gaussian $\beta(t)$ and a Gamma-mixture $\gamma(t)$. Therefore the cumulant-generating
function, {\it cf.} \Sec{approx}, of $\alpha(t)$ is a sum of the Gaussian
cumulant-generating function \cite[Chapter 10]{Krishnamoorthy} and the Gamma-mixture
distribution, obtained in Appendix~\ref{sec:lemma}. This allows us to calculate the
mean, variance ($\var{\alpha(t)}$) and the skewness ($\skw{\alpha(t)}$) of
$\alpha(t)$. Omitting straightforward computational details, we write immediately
the final results for the steady-state solution $\alpha_\mathrm{SS}$
\begin{eqnarray}\EQ{avg}
  \avg{\alpha_\mathrm{SS}} &=& \lim_{t\to\infty} \{\kappa_1[\beta(t)]+\kappa_1[\gamma(t)]\} =  \frac{C}{\tau A}
  \\\EQ{var}
  \var{\alpha_\mathrm{SS}} &=& \lim_{t\to\infty} \{\kappa_2[\beta(t)]+\kappa_2[\gamma(t)]\} \nonumber\\
		&=& \frac{B^2 + C^2/\tau}{2 A}
  \\\EQ{skw}
  \skw{\alpha_\mathrm{SS}} &=& \lim_{t\to\infty} \frac{\kappa_3[\beta(t)]+\kappa_3[\gamma(t)]}{\var{\alpha(t)}^{3/2}}
      \\&=& \frac{4 \sqrt{2 A} C^3/\tau}{3 (B^2 + C^2/\tau)^{3/2}}\text{,}
\end{eqnarray}
where $\kappa_i$ stands for the $i$-th cumulant.

Since the skewness of the steady-state solution does not vanish for $C > 0$, {\it cf.} \Eq{skw},
the structure of the fluctuations is non-Gaussian. Moreover, the skewness has the
same sign as $C$, which is consistent with the experimental observations of Ref.~\cite{PRE2016}.
Higher order statistics, than those in Eqs.~(\ref{eq:avg})-(\ref{eq:skw}), can also
be computed from the cumulant-generating function.

The physical meaning of the parameters $A$, $B$ and $C$ depends on the problem,
modeled by the extended Langevin dynamics. For instance, for a Brownian particle,
$A$ is the friction constant, while the parameter $B$ can be computed from the
kinetic temperature \footnote{
  The kinetic temperature is proportional to the variance of the particle's linear
  momentum distribution, {\it cf.} \Eq{var}.
} once $A$ and $C$ are known. Finally, as explained further, $C/\tau$ is the average
magnitude of the external force. This can be seen, if we take the steady-state
average of both sides in \Eq{langevin}, which corresponds to the macroscopic
dynamics:
\begin{equation}\EQ{avg1}
  \avg{d\alpha(t)/dt} = -A \avg{\alpha(t)} + \avg{\epsilon(t)} = 0\text{,}
\end{equation}
where the left hand side must vanish in the steady state by definition. From Eqs.~(\ref{eq:super})
and (\ref{eq:avg1}) one finds:
\begin{eqnarray}\EQ{avg2}
  \avg{\alpha(t)} &=& A^{-1}\avg{B dW(t)/dt + C dE_\tau(t)/d(t)} \nonumber\\
		&=& A^{-1}\avg{C dE_\tau(t)/dt}\text{,}
\end{eqnarray}
since the average effect of the white noise, $dW(t)$ vanishes. Finally, $\avg{C dE_\tau(t)/dt} = C/\tau$,
because
\begin{equation}\EQ{avg3} d\avg{C \int_0^t dE_\tau(t)}/dt = d(C t/\tau)/dt = C/\tau\text{,}\end{equation}
due to Eqs.~(\ref{eq:spec})-(\ref{eq:a2}).

Combining Eqs.~(\ref{eq:avg1})-(\ref{eq:avg3}), we have
$$\avg{\alpha(t)} = \frac{C}{\tau A}\text{,}$$
which relates the terminal value $\avg{\alpha(t)}$ to the external force $C/\tau$
and the friction coefficient $A$. For the complete description of the extended Langevin
dynamics, the time scale parameter $\tau$, which is a new characteristic of a system,
needs to be found as well.

In other words, all constants $A$, $B$ and $C/\tau$ are physical observables, which
can be determined by measurements. In fact, these quantities are used to characterize
physical systems in steady states, as was shown in the example of a Brownian motion
above.

\SEC{end}{Conclusion}
The Langevin dynamics of \Eq{dff} was extended by a new term to obtain \Eq{ours},
which generalizes the theory of Brownian motion, as well as the Onsager-Machlup
theory of fluctuations, from spatially symmetric equilibrium systems to equilibrium
and nonequilibrium steady-state systems with a preferred spatial direction. We also
provided statistical arguments in \Sec{gamma}, which allowed us to deduce the form of
the new term.

A method of solving the extended Langevin equation was demonstrated in \Sec{new}.
In particular, we showed how the statistical properties of its steady-state solution
can be computed from macroscopic physical observables. The steady-state probability
distribution of the fluctuations is also characterized by the cumulant-generating
function, which can be expressed using the dilogarithm, a special mathematical function,
{\it cf.} Appendix~\ref{sec:lemma}. The corresponding probability density function,
which apparently can not be expressed in terms of elementary functions, can be in
principle approximated by the Modulated Gaussian distribution \cite{PRE2016} for
practical applications.

The theory, presented in this paper, should be applicable to a variety of physical
systems in classical statistical mechanics, such as an equilibrium fluid system
in a gravitational potential or an electric current driven by a voltage difference.
Applications of \Eq{ours} to particular systems opens new perspectives for the future
research in equilibrium and nonequilibrium statistical physics.

Finally, we would like to make a remark about the equilibrium systems in the magnetic
field. The vector of a magnetic field has an {\it axial} nature, which means that
it does not select a preferred direction, but rather determines a sense of rotation
in its normal plane. As discussed in \Sec{intro}, such systems have a symmetry of
the Curie's limiting point group $\infty / m$, which does not admit a preferred
spatial direction \cite{Shubnikov1988}. This is in contrast to the forces,
described by {\it polar} vectors considered here, {\it e.g.} the electric field,
which belongs to the symmetry group $\infty \cdot m$ \cite{Shubnikov1988}. For
this reason, systems subject to a magnetic field bear more similarity with the
rotating systems, which still may need a further generalization of the Langevin
equation.

\appendix
\SEC{approx}{Gamma approximation of the binomial distribution}
In this Appendix we provide a formal mathematical argument for the Gamma approximation
\Eq{gamma} of the binomial distribution \Eq{binomial}. For this we will compare
a {\it cumulant-generating function} \cite[Section 26.1]{Berberan-Santos2006,AbramowitzStegun}
of the Gamma distribution ($\mathcal{C}_\Gamma)$ with that of the Binomial distribution
($\mathcal{C}_B$).

We recall that a probability distribution of a random variable $m$ is uniquely
determined by its probability mass (or density) function or, equivalently, by its
cumulant-generating function
$\mathcal{C}(k)$:
\begin{equation}\EQ{cgf}
  \mathcal{C}(k) = \ln \avg{\exp(k m)}_m = \sum_{i=1}^{\infty} \kappa_j \frac{k^j}{j!}\text{.}
\end{equation}
where $k$ is the dual of $m$, while the angle brackets denote the average value.
The Taylor coefficients $\kappa_j$ in \Eq{cgf} are the {\it cumulants} of $m$.

The first and second cumulants of a probability distribution are equal to its mean
and its variance, respectively. The mean and the variance of the binomial distribution
\Eq{binomial} are equal to those of the Gamma distribution \Eq{gamma}, respectively,
if $p(1) = \theta = 1/2$, {\it cf.} \cite[Chapter 3 and 15]{Krishnamoorthy}. It
follows then, that their cumulant-generating functions agree up to the third order
term in $k$, {\it i.e.}
$$\mathcal{C}_B(k) - \mathcal{C}_\Gamma(k) = \mathcal{O}(k^3)\text{,}$$
because the first two cumulants cancel each other in the series expansion \Eq{cgf}
for $\mathcal{C}_B(k)$ and $\mathcal{C}_\Gamma(k)$, respectively.

For the binomial distribution \Eq{binomial}, there are no asymptotic formulae of
an accuracy higher than $\mathcal{O}(k^3)$ with the support on the half real line
\footnote{
  However, there may exist approximations of \Eq{binomial} with the same order of
  accuracy as \Eq{gamma}.
}. The error of the third order is due to the skewness of the Gamma distribution.
This property is inherent in all distributions, which have the support on the half
real line, as a consequence of their obvious asymmetry. In other words, any asymptotic
formula of \Eq{binomial} acquires skewness in the limit $N\to\infty$, if its support
spreads over all non-negative reals $[0,N)_{N\to\infty}$, as considered in \Sec{gamma}.

\SEC{lemma}{Gamma-mixture probability distribution}
When solving the extended Langevin equation in \Sec{new}, we had to evaluate a
steady-state limit for a definite stochastic integral of the form:
\begin{equation}\EQ{I}
  I = \int_0^t dE_\tau(s) \phi(s)\text{.}
\end{equation}
where $\phi(s) = C \exp[A(s-t)]$ and $dE_\tau(s)$ is the exponential noise with
the time scale parameter $\tau$.

Below we will consider a more general function $\phi(s)$. We will obtain the
cumulant-generating function of the random variable $I$ \cite[Section 26.1]{Berberan-Santos2006,AbramowitzStegun},
{\it cf.} \Sec{approx}, and compute some of its statistical moments, {\it i.e.}
mean, variance, and skewness.

\begin{lemma}\label{lemma}
    Let $I$ be a random variable given by $$I = \int_0^t dE_\tau(s) \phi(s)\text{,}$$
    where $\phi(s)$ is some integrable function and $dE_\tau(s)$ is the exponential noise
		with the time scale parameter $\tau$.	Then $I$ has a Gamma-mixture distribution,
		which is described by a cumulant-generating function
    $$\mathcal{C}(\tilde{I}) = -\int_0^t \frac{ds}{\tau} \ln[1-\phi(s)\tilde{I}]\text{,}$$
    where $\tilde{I}$ is the dual of $I$ in the reciprocal Laplace space.
\end{lemma}
\begin{proof}[Proof]
Partitioning the domain of integration $[0,t]$ into $n$ subintervals of length
$\Delta{t}$, so that $n \Delta{t} = t$, we express $I$ as the limit of the
following discrete sum $S$:
\begin{equation}\EQ{sum}
I \underset{n\to\infty}{\to} S_n = \sum_{j=0}^{n-1} \phi(j\Delta{t})
  \int_{j\Delta{t}}^{(j+1)\Delta{t}} dE_\tau(s) = \sum_{j=0}^{n-1} r_j
\text{,}
\end{equation}
where the index $j$ runs through all subintervals.

By virtue of Eqs.~(\ref{eq:a1}-\ref{eq:a2}), each term of the summation $r_j$ in
\Eq{sum} is an independent Gamma-distributed random variable. In other words, the
probability distribution of $I$ is a discrete {\it mixture of Gamma-distributed}
random variables or, equivalently, a discrete Gamma-mixture distribution. The shape
and scale parameters \cite[Chapter 15]{Krishnamoorthy} of each component $r_j$
are, respectively, $\Delta{t}/\tau$ and $\phi(j\Delta{t})$, while the cumulant-generating
function of their sum is:
\begin{equation}\EQ{GS}
  \mathcal{C}(\tilde{S}_n) = -\sum_{j=0}^{n-1} \frac{\Delta{t}}{\tau} \ln[1-\phi(j\Delta{t}) \tilde{S}_n]\text{,}
\end{equation}
where $\tilde{S}_n$ is the dual of $S_n$.

One should recognize in \Eq{GS} a Riemann sum, which in the limit $\Delta{t}\to0$
($n\to\infty$) becomes an integral. Then from Eqs.~(\ref{eq:sum}) and (\ref{eq:GS})
we conclude that
\begin{equation}\EQ{cgfI}
  \mathcal{C}(\tilde{I})
    = \lim_{n\to\infty}\mathcal{C}(\tilde{S}_n)
    = -\int_0^t \frac{ds}{\tau} \ln[1 - \phi(s) \tilde{I}]\text{,}
\end{equation}
which is the cumulant-generating function of the Gamma-mixture distribution.
\end{proof}

The cumulants $\kappa_i(I)$, and hence the statistical moments of $I$, can be obtained
either by differentiation of the cumulant-generating function given by Lemma~\ref{lemma},
{\it cf.} \Eq{cgf}, or by using the calculus of cumulants. While the latter method
was adopted in \Sec{new} to compute the skewness of the steady-state solution
$\alpha_\mathrm{SS}$, {\it cf.} \Eq{skw}, in this section the former approach is
more convenient.

Differentiating \Eq{cgfI} with respect to $\tilde{I}$ we find:
\begin{eqnarray}\EQ{momI}
  \avg{I} &=& \kappa_1(I) = \left. \frac{d\mathcal{C}(\tilde{I})}{d\tilde{I}}\right|_{\tilde{I}=0}
      = \int_0^t \frac{ds}{\tau} \phi(s)\nonumber\\
  \var{I} &=& \kappa_2(I) = \left. \frac{d^2\mathcal{C}(\tilde{I})}{d\tilde{I}^2}\right|_{\tilde{I}=0}
      = \int_0^t \frac{ds}{\tau} \phi(s)^2 \nonumber\\
  \kappa_3(I) &=& \left. \frac{d^3\mathcal{C}(\tilde{I})}{d\tilde{I}^3}\right|_{\tilde{I}=0}
      = 2 \int_0^t \frac{ds}{\tau} \phi(s)^3 \text{,}
\end{eqnarray}
from which the skewness can be computed using its definition in terms of cumulants
$\skw{I}=\kappa_3(I)/\kappa_2(I)^{3/2}$.

Now, returning to \Eq{I}, we need to use $\phi(s) = C \exp[A(s-t)]$ in Eqs.~(\ref{eq:cgfI})-(\ref{eq:momI}).
Then the cumulant-generating function, given by \Eq{cgfI}, becomes \footnote{
  We evaluated the integral in \Eq{cgfI} for $\phi(s) = C \exp[A(s-t)]$ using
  a software for symbolic computations \cite{Wolfram}.
}:
\begin{equation}\EQ{poly1}
  \mathcal{C}(\tilde{I}) = \frac{1}{\tau A} \left\{
    \Li_2(C \tilde{I})-\Li_2[C \tilde{I} \exp(-A t)]
  \right\}\text{,}
\end{equation}
where $\Li_2$ stands for a dilogarithm function \cite[Section 27.7]{AbramowitzStegun}.
So that the steady-state limit of \Eq{poly1} yields
\begin{equation}\EQ{poly2}
  \lim_{t\to\infty} \mathcal{C}(\tilde{I}) = \Li_2(C \tilde{I})/ (\tau A)\text{.}
\end{equation}

To evaluate \Eq{momI} for the form of $\phi(s)$, chosen above, it is convenient
to consider first a general integral of the following form
\begin{eqnarray}\EQ{n}
  \int_0^t \frac{ds}{\tau} \phi(s)^n &=& \int_0^t \frac{ds}{\tau} C^n \exp[n A(s-t)] \nonumber\\
		&=& \frac{C^n}{n \tau A}[1 - \exp(-n A t)]\text{,}
\end{eqnarray}
for any integer $n$.

Using then \Eq{n}, we can calculate the statistics given in \Eq{momI}, as well as
their steady-state limits:
\begin{eqnarray}\EQ{momSS}
  \lim_{t\to\infty} \avg{I} = \lim_{t\to\infty}\left\{ \frac{C}{\tau A} [1-\exp(-A t)] \right\} = \frac{C}{\tau A} \nonumber\\
  \lim_{t\to\infty} \var{I} = \lim_{t\to\infty}\left\{ \frac{C^2}{2\tau A} [1-\exp(-2 A t)] \right\} = \frac{C^2}{2\tau A}\nonumber\\
  \lim_{t\to\infty} \kappa_3 = \lim_{t\to\infty}\left\{ \frac{2 C^3}{3\tau A} [1-\exp(-3 A t)] \right\} = \frac{2 C^3}{3\tau A}\text{.}\nonumber\\
\end{eqnarray}

\bibliographystyle{apsrev4-1}
\bibliography{References}

\end{document}